\documentclass[10pt,journal]{IEEEtran}

\usepackage{times}
\usepackage{amsmath,amssymb}
\usepackage{graphicx}
\usepackage{cite}
\usepackage{url}
\usepackage{graphicx}

\usepackage{xcolor}
\usepackage{tcolorbox}

\newtcolorbox{boxedtext}{
  colback=white,
  colframe=blue,
  boxrule=0.8pt,
  arc=2mm
}

\usepackage{amsthm}

\newtheorem{theorem}{Theorem}
\newtheorem{proposition}{Proposition}
\newtheorem{corollary}{Corollary}

\theoremstyle{definition}

\usepackage{enumitem}

\usepackage{hyperref}
\usepackage{xurl}
\title{From Signals to Causes: A Causal Signal Processing Framework for Robust and Interpretable Clinical Risk Prediction}

\author{Surajit~Das and~Maxine~Tan%
\thanks{Surajit Das is with the Lab Infochemistry Scientific Center, ITMO University, St. Petersburg, Russia (e-mail: mr.surajitdas@gmail.com). ORCID: 0009-0008-6692-6697.}
\thanks{Maxine Tan is with the School of Engineering, Monash University, Malaysia (e-mail: maxine.tan@monash.edu). ORCID: 0000-0001-5071-2477.}
}

\begin{document}
\maketitle

\begin{abstract}
Learning-based signal processing systems increasingly support high-stakes medical decisions using heterogeneous biomedical signals, including medical images, physiological time series, and clinical records. Despite strong predictive performance, many models rely on statistical correlations that are unstable across acquisition settings, patient populations, and institutional practices, limiting robustness, interpretability, and clinical trust.

We advocate a causal signal processing perspective in which biomedical signals are treated as effects of latent generative mechanisms rather than as isolated predictive inputs. Using clinical risk prediction as a motivating example, we show how disease-related factors generate observable biomarkers, while acquisition processes act as confounders influencing signal appearance. In clinical disease risk prediction from chest CT scans and patient risk factors, correlational models may fail under scanner changes, whereas causal abstractions remain invariant.

Building on this view, we propose a unifying conceptual framework integrating causal modeling with learning-based signal processing and neuro-symbolic reasoning. Statistical models extract multimodal representations that are mapped to interpretable causal abstractions and combined with symbolic knowledge encoding clinical risk factors and guidelines. This structure enables clinically grounded explanations, counterfactual reasoning about hypothetical interventions, and improved robustness to distribution shifts arising from changes in acquisition conditions or screening policies.

Rather than introducing a specific algorithm, this article presents schematic causal structures and a comparative analysis of correlation-based, causal, and neuro-symbolic approaches to guide the design of robust and interpretable medical decision-support systems.
\end{abstract}

\section{Introduction}

Learning-based signal processing systems increasingly support high-stakes decision making in medical applications, where models operate on heterogeneous biomedical signals such as medical images, physiological time series, and clinical records \cite{wang2024multimodal, sun2025causal, Tegomoh}. These systems achieve impressive predictive accuracy by exploiting statistical regularities in historical data \cite{wang2024multimodal}. However, such regularities are often specific to particular acquisition settings, patient populations, and institutional practices \cite{han2025distribution}. As a result, models trained on correlational patterns may fail under deployment, exhibit limited interpretability, and provide explanations that are misaligned with clinical reasoning \cite{Tegomoh}.

In medical environments, changes in scanners, imaging protocols, screening policies, and population demographics are common \cite{han2025distribution}. These changes induce distribution shifts that are frequently interpreted as statistical nuisances but, from a causal perspective \cite{Castro1}, correspond to interventions on the data-generating process \cite{bing2023invariance, han2025distribution}. Popular interpretability tools, such as saliency maps and feature attributions, indicate which inputs influence predictions but do not explain why those inputs matter clinically \cite{sokol2025clinical}. This gap between statistical prediction and clinical reasoning limits trust, accountability, and generalization in real-world settings\cite{vangeloven2020prediction, sanchez2022causal}.

Causal inference \cite{Prosperi, vangeloven2020prediction, sanchez2022causal} provides a principled framework for reasoning about data-generating mechanisms and interventions \cite{bing2023invariance}. From a signal processing perspective, a causal view treats biomedical observations as outputs of latent generative processes governed by disease mechanisms, patient characteristics, and acquisition conditions \cite{sun2025causal}. Robust prediction then corresponds to learning relationships that remain invariant under interventions on non-causal variables, while interpretability corresponds to explanations grounded in causal structure rather than in spurious correlations \cite{bing2023invariance, Pearl2009}.

This article reformulates clinical risk prediction as a causal signal processing problem \cite{sun2025causal}. Multimodal biomedical signals are modeled as mixtures of latent causal sources and acquisition channels, analogous to source separation and channel modeling in classical signal processing. Latent disease states act as hidden generators of observable biomarkers, while measurement processes transform these generators into recorded signals. Within this abstraction, the objective shifts from learning a direct mapping between observations and labels to estimating latent causal states that govern both signal appearance and clinical outcomes. Fig.~ \ref{fig:conceptual_comparison} illustrates the conceptual distinction between conventional correlation-based learning and the proposed causal signal processing perspective, highlighting how acquisition-related factors introduce spurious associations that causal abstractions aim to remove

\paragraph*{Scope and Contribution.}
This article provides a tutorial and unifying signal-processing reformulation of causal modeling for clinical risk prediction. Rather than proposing a new algorithm, we synthesize ideas from causal inference, representation learning, and multimodal signal processing into a coherent signal-to-decision framework. The framework clarifies how robustness to distribution shift \cite{felekis2026distributionally}, clinically grounded interpretability, and generalization under intervention can be addressed through causal abstraction and invariant representation learning. By bridging causal modeling with core signal processing concepts such as source separation, channel modeling, and invariance, this work provides design principles for building intervention-aware medical decision-support systems. This article emphasizes conceptual unification rather than algorithmic development or empirical benchmarking.

\paragraph*{Neuro-Symbolic Abstraction.}
Operationalizing causal structure requires explicit representations of clinically meaningful variables such as risk factors, latent disease states, and intervention targets. These abstractions are structured and governed by domain knowledge and clinical guidelines, and cannot be reliably inferred or manipulated using statistical learning alone. Neuro-symbolic reasoning therefore plays a functional role in the proposed framework: statistical models extract uncertain evidence from raw signals, while symbolic representations encode causal relations among abstractions and support counterfactual and intervention-based reasoning. This separation is essential for producing explanations that are stable under changes in acquisition conditions and aligned with clinical decision-making. 

\paragraph*{Article Organization.}
We first review background concepts from causal inference and their relevance to signal processing. We then describe causal structures underlying multimodal clinical risk prediction and interpret distribution shifts as interventions on data-generating mechanisms. Next, we examine implications for robustness and interpretability and introduce neuro-symbolic reasoning as a means of causal abstraction. Finally, we discuss methodological implications and open challenges for causal signal processing in medical decision support.

\begin{figure}[t]
    \centering
    \includegraphics[width=\linewidth]{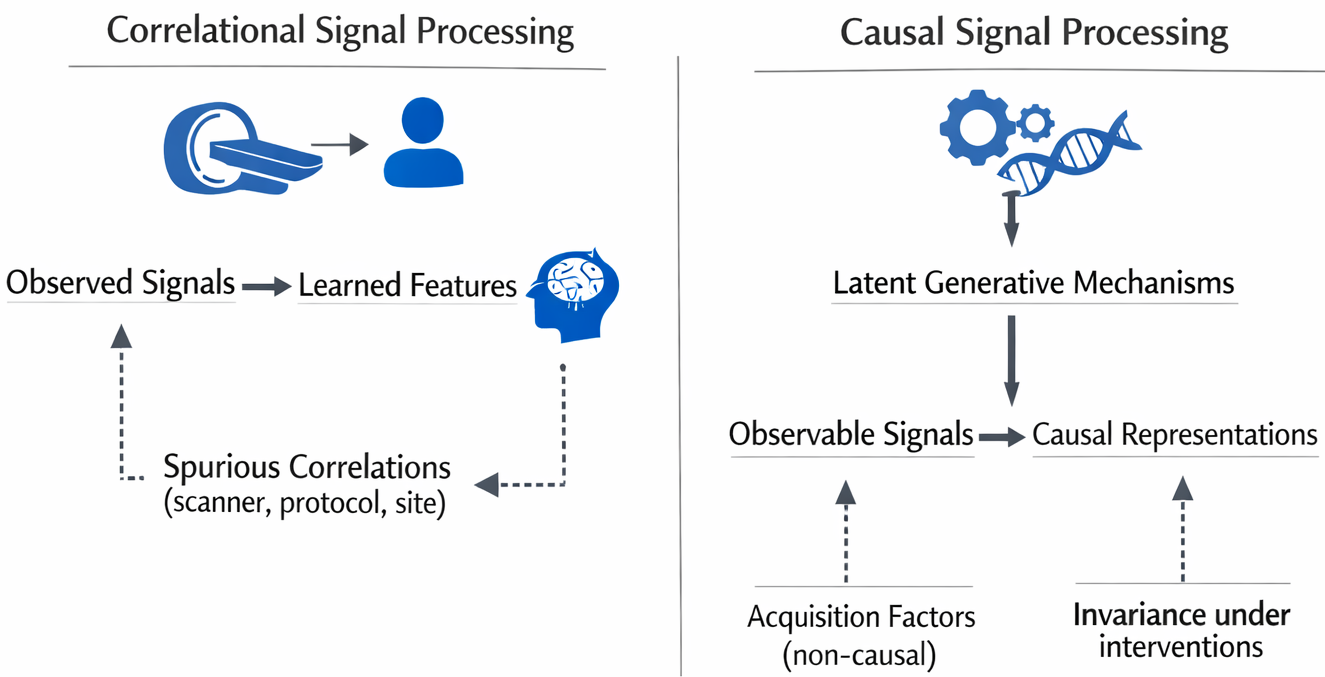}
    \caption{Conceptual comparison between correlational and causal signal processing paradigms. 
Conventional learning-based systems exploit statistical associations in observed signals, 
including spurious correlations induced by acquisition and institutional factors. 
Causal signal processing models signals as outcomes of latent generative mechanisms, 
seeking representations that remain invariant under interventions on non-causal variables.}

    \label{fig:conceptual_comparison}
\end{figure}

%@@@@@@@@@@@@@@@@@@@@@@@@@@@@@@@@@@@@@@@@@@@@@@@@@@

\section{Background: Clinical Risk Prediction and Signal Processing}
\label{sec:background}

Clinical risk prediction aims to estimate the likelihood of future disease outcomes using heterogeneous biomedical data, where heterogeneous biomedical observations are understood as measurements generated by a hierarchy of hidden physiological and pathological processes. This perspective provides a principled decomposition:

\begin{itemize}
    \item \textbf{Observed Signals:} The raw, multi-modal data, including medical images, physiological time series, laboratory measurements, and electronic health records.
    \item \textbf{Latent States:} The unobserved causal variables of interest, such as the true disease state, its progression stage, and the patient's underlying physiological risk profile.
    \item \textbf{Clinical Decisions:} The actionable outputs informed by the estimated latent state, including diagnostic calls, screening strategies, and treatment selections.
\end{itemize}

Recent advances in machine learning have enabled powerful data-driven models that map such observations directly to risk scores or diagnostic labels. These models are typically trained to optimize predictive accuracy on historical datasets and are evaluated under the assumption that training and deployment data are drawn from the same distribution.

From a signal processing perspective, this paradigm treats biomedical observations as input signals and clinical outcomes as target variables, focusing on statistical pattern extraction. While effective within controlled datasets, this formulation overlooks the structured nature of variability in medical data. Observed signals are shaped not only by underlying biological processes but also by acquisition devices, protocols, and institutional practices. As a result, predictive patterns learned from data may reflect confounding influences rather than true disease mechanisms.

\subsection{Conventional Learning-Based Formulation}

In most contemporary systems, clinical risk prediction is formulated as a supervised learning problem:
\[
f: \mathbf{x} \rightarrow y,
\]
where $\mathbf{x}$ denotes multimodal observations (e.g., images and clinical variables) and $y$ denotes a clinical outcome or risk score. Model training emphasizes minimizing empirical error under the assumption that correlations present in the training set will persist at deployment.

This formulation implicitly assumes that all predictive structure in $\mathbf{x}$ is equally informative. Consequently, models may exploit correlations induced by scanner type, site-specific protocols, or population composition, even when these factors are unrelated to disease etiology. Such dependencies are often invisible during training and only emerge as failures when models are deployed in new environments.

\subsection{Sources of Variability in Medical Signals}

Critically, the variability in clinical data is not merely stochastic noise; it is highly structured and arises from multiple, distinct causal sources:
\begin{itemize}
    \item \textbf{Patient-Level Sources (The Primary Source):} Genetics, age, sex, and lifestyle, which directly modulate the underlying pathological process.
    \item \textbf{Measurement/ Acquisition Sources (The Channel):} Scanner manufacturer, imaging protocol, operator technique, and laboratory assay batch, which transform the true biological signal without altering its source.
    \item \textbf{Clinical History Sources (Temporal Confounding):} Prior treatments, screening history, and disease course, which alter the presentation of the current latent state.
\end{itemize}

Variability in biomedical signals arises from multiple distinct sources, as discussed below:

\begin{itemize}
    \item \textbf{Biological variability}, reflecting true differences in disease state, physiology, and patient risk factors. It is associated with Patient-Level Sources.
    \item \textbf{Measurement variability}, induced by imaging devices, acquisition parameters, and operator practices. It arises from Measurement / Acquisition Sources.
    \item \textbf{Contextual variability}, including prior interventions, screening history, and institutional policies. It arises from clinical history sources.
\end{itemize}

Standard learning-based approaches treat these sources as undifferentiated noise or nuisance variation. However, they have fundamentally different interpretations: biological variability carries causal information about disease, whereas measurement and contextual variability do not. Failure to distinguish between these sources leads to models that are brittle under distribution shift and difficult to interpret clinically.

\subsection{Limitations of Correlation-Based Risk Models}

Correlation-based risk models suffer from three key limitations in clinical settings:

\begin{itemize}
    \item \textbf{Limited robustness:} Performance degrades under changes in scanners, protocols, or patient populations.
    \item \textbf{Weak interpretability:} Explanations identify influential signal patterns but do not clarify their clinical meaning.
    \item \textbf{Poor alignment with clinical reasoning:} Predictions cannot be reliably assessed under hypothetical interventions (e.g., altered patient behavior or screening policy).
\end{itemize}

These limitations arise because the dominant modeling paradigm conflates disease-related variation with acquisition-related and institutional variation. As a result, learned representations may be predictive without being causally meaningful.

\subsection{Motivation for a Causal Signal Processing Perspective}

Signal processing has a long tradition of modeling observed data as mixtures of latent sources transmitted through channels. This perspective naturally suggests a reinterpretation of clinical risk prediction: biomedical signals can be viewed as observations generated by latent physiological and pathological processes and transformed by acquisition channels.

Such a formulation motivates shifting the learning objective from direct prediction to \textbf{latent state estimation}, where the goal is to recover disease-related variables that govern both observable signals and clinical outcomes. This shift provides the conceptual foundation for incorporating causal structure into representation learning, enabling models that prioritize invariance to non-causal variability and support clinically meaningful interpretation.

This background motivates the causal signal processing framework developed in the following sections, which reformulates clinical risk prediction in terms of latent causal sources, structured variability, and intervention-aware inference.

%@@@@@@@@@@@@@@@@@@@@@@@@@@@@@@@@@@@@@@@@@@@@@@@@@@

%@@@@@@@@@@@------3------@@@@@@@@@@@@@@@@@@@@@@@@@@@

\section{A Causal View of Multimodal Medical Signals}
\label{sec:causal_view}

Medical data are generated by structured causal mechanisms rather than independent statistical processes. A causal view models how biological, behavioral, and environmental factors produce observable measurements and clinical outcomes \cite{Pearl2009}. This perspective enables models that generalize across populations, acquisition settings, and clinical interventions.

\subsection{Causal Roles in Medical Data}

In clinical risk prediction, variables assume distinct causal roles:
\begin{itemize}
    \item \textbf{Causal factors:} variables that directly influence disease risk (e.g., smoking, genetics, age).
    \item \textbf{Mediators:} observable manifestations of disease state (e.g., imaging biomarkers, physiological measurements).
    \item \textbf{Confounders:} factors that affect observations without altering disease (e.g., scanner type, protocols, institutional practices).
\end{itemize}

Conventional learning models treat all variables symmetrically and may exploit spurious associations induced by confounders. Causal modeling distinguishes these roles explicitly, clarifying which variables should influence predictions under intervention.

\paragraph{Minimal Structural Causal Model.}
Formally, we represent clinical risk prediction using a structural causal model (SCM)
with latent disease state $D$, causal risk factors $R$, acquisition-related confounders $A$,
observed biomedical signals $X$, and clinical outcome $Y$:
\begin{align}
D &= f_D(R, U_D), \\
X &= f_X(D, A, U_X), \\
Y &= f_Y(D, U_Y),
\end{align}
where $U_D, U_X, U_Y$ denote mutually independent exogenous noise variables.
This structure encodes the assumption that acquisition and institutional factors
affect signal appearance without altering the underlying disease state.

\begin{boxedtext}
\textbf{Box 1: Key Concepts in Causal Signal Processing}
\begin{itemize}
    \item \textbf{Signals as effects:} Observed signals arise from latent generative mechanisms.
    \item \textbf{Latent causal variables:} Disease states act as hidden sources generating observable patterns.
    \item \textbf{Confounding factors:} Acquisition processes affect signals without affecting disease.
    \item \textbf{Interventions:} Changes in specific mechanisms of the data-generating process.
    \item \textbf{Causal invariance:} Stable relationships under non-causal interventions generalize.
\end{itemize}
\end{boxedtext}

\subsection{Causal Structure of Multimodal Risk Prediction}

Lung cancer risk prediction illustrates this structure: behavioral and biological risk factors influence a latent disease state, which generates imaging biomarkers and clinical findings. Acquisition and institutional variables affect signal appearance without causing disease. Represented as a directed acyclic graph, this structure implies that predictions should depend on latent disease variables rather than acquisition artifacts, and that robustness requires invariance to non-causal variables.

\paragraph{Signal Processing View.}
From a signal processing perspective, the observation model can be written as
\begin{equation}
X = g_A\!\left(h(D)\right) + \varepsilon,
\end{equation}
where $h(D)$ denotes a disease-dependent latent signal source,
$g_A(\cdot)$ represents an acquisition-dependent channel operator,
and $\varepsilon$ captures measurement noise.
Robust prediction requires learning representations that are invariant to $g_A$
while preserving information about $D$.

\subsection{Distribution Shift as Intervention}

In medical AI, interventions include changes in imaging devices, screening policies, and deployment populations. Distribution shifts arise from such interventions. Models based on non-causal correlations are sensitive to these changes, whereas causal models seek invariant relationships between latent disease states and clinically meaningful signals.

\subsection{Beyond Correlational Interpretability}

Saliency maps and feature attribution methods indicate which inputs influence predictions but not why those inputs are clinically relevant \cite{Selvaraju2017}. Because they are correlational, they remain sensitive to confounders.

Clinically meaningful explanations must support counterfactual reasoning, such as whether predictions would change if a patient had not smoked or if acquisition protocols differed. Such questions require explicit causal structure rather than statistical importance alone.

\paragraph{Counterfactual Interpretation.}
Changes in causal risk factors $R$ induce corresponding changes in the latent disease state and therefore in the predicted clinical outcome. This relationship can be expressed as
\begin{equation}
Y(r') = f_Y\bigl(f_D(r', U_D), U_Y\bigr),
\end{equation}
where $r'$ denotes a hypothetical modification of patient risk factors.

In contrast, changes in acquisition or institutional variables affect the appearance of observed signals but do not alter the underlying disease mechanism, and therefore do not affect the predicted outcome:
\begin{equation}
Y(a') = Y.
\end{equation}

This distinction formalizes why clinically meaningful explanations must respond to disease-related changes while remaining invariant to variations in acquisition conditions. Causal interpretability is thus evaluated by assessing whether predictions vary appropriately under disease-related changes and remain stable under non-causal perturbations.

%-----------------------------------------------
\iffalse

\paragraph{Counterfactual Interpretation.}
Counterfactual predictions under causal interventions on risk factors $R$ can be expressed as
\begin{equation}
Y_{do(R=r')} = f_Y\!\left(f_D(r', U_D), U_Y\right),
\end{equation}
whereas interventions on acquisition variables satisfy
\begin{equation}
Y_{do(A=a')} = Y.
\end{equation}
This distinction formalizes why clinically meaningful explanations must respond to
causal changes while remaining invariant to acquisition effects.

\fi

%----------------------------------------------

Causal interpretability evaluates whether predictions change under interventions on causal variables and remain invariant under acquisition changes \cite{Samek2021}.

\begin{boxedtext}
\textbf{Box 2: Why Causal Reasoning Matters in Practice}
\begin{itemize}
    \item \textbf{Robustness:} Stability under acquisition changes.
    \item \textbf{Interpretability:} Explanations reflect disease mechanisms.
    \item \textbf{Counterfactual reasoning:} Supports clinically relevant ``what-if'' queries.
    \item \textbf{Trust and accountability:} Makes assumptions explicit for regulation and oversight.
\end{itemize}
\end{boxedtext}

\subsection{Implications for Signal Processing Methodology}

A causal view implies:
\begin{itemize}
    \item \textbf{Representation learning} should emphasize invariance to non-causal variability.
    \item \textbf{Multimodal fusion} should respect causal roles of modalities.
    \item \textbf{Evaluation} should assess stability under simulated or real interventions.
\end{itemize}

Explicit modeling of signal generation and intervention effects provides a principled path toward robust and clinically meaningful AI systems.

%======================================

%===================================

%@@@@@@@@@@@@@@@@@@@@@@@----4----@@@@@@@@@@@@@@@@@@@

\section{Interventions, Distribution Shifts, and Robustness}
\label{sec:robustness}

Clinical decision-support systems operate in evolving environments. Imaging protocols, populations, and screening policies change over time. From a causal perspective, these changes correspond to interventions on the data-generating process.

\subsection{Interventions in Medical Data Generation}

In medical signal processing, interventions affect:
\begin{itemize}
    \item \textbf{acquisition processes} (scanner hardware, imaging parameters),
    \item \textbf{population composition} (demographics, disease prevalence),
    \item \textbf{clinical policies} (screening thresholds, follow-up strategies).
\end{itemize}
Such interventions modify observed signals without altering biological disease mechanisms.

\subsection{Distribution Shift as a Causal Phenomenon}

Distribution shift reflects changes in causal mechanisms rather than purely statistical variation. Shifts arise from changes in confounder distributions, measurement processes, and institutional policies. Correlation-based models may rely on unstable associations, whereas causal models isolate invariant relationships between latent disease states and meaningful signals.

\subsection{Robustness Through Causal Invariance}

Robust models remain invariant to scanner type and acquisition settings while responding appropriately to disease-related changes. 

\paragraph{Invariant Representation Condition.}
Let $Z = \phi(X)$ denote a learned representation.
Causal robustness corresponds to the conditional independence
\begin{equation}
Z \;\perp\!\!\!\perp\; A \mid D,
\end{equation}
and equivalently,
\begin{equation}
Y \;\perp\!\!\!\perp\; A \mid Z.
\end{equation}
These conditions ensure that predictions depend on disease-related information
and remain invariant under interventions on acquisition variables.

Signal processing methods that promote invariance through structured representations or causal regularization improve generalization without exhaustive retraining.

\subsection{Implications for Training and Evaluation}

A causal view of robustness implies:
\begin{itemize}
    \item data augmentation should simulate plausible interventions,
    \item cross-dataset validation approximates interventional testing,
    \item uncertainty estimates should reflect causal ambiguity.
\end{itemize}
Robustness is thus assessed by stability under clinically meaningful change rather than accuracy alone.

\paragraph{Invariant Risk Objective.}
Given multiple environments $e, e' \in \mathcal{E}$ corresponding to different acquisition
conditions or institutions, invariant learning can be expressed as
\begin{align}
\min_{\phi,h} \sum_{e \in \mathcal{E}}
\mathbb{E}_{(X,Y)\sim P_e}
\big[\ell(h(\phi(X)), Y)\big] \\
\quad \text{s.t.} \quad
P_e(Y \mid \phi(X)) = P_{e'}(Y \mid \phi(X)),
\end{align}
for all $e, e' \in \mathcal{E}$.
This constraint enforces predictive relationships that are stable under non-causal interventions.

%@@@@@@@@@@@@@@@------5------@@@@@@@@@@@@@@@@@@@@

\section{Interpretable Risk Modeling Beyond Saliency Maps}
\label{sec:interpretability}

Interpretability is essential for clinical deployment but is often reduced to post hoc visualizations. These methods indicate which inputs influence predictions but do not explain decisions in terms of disease mechanisms and remain correlational.

From a causal perspective, interpretability must be grounded in explanatory variables corresponding to clinically meaningful causes.

\subsection{Limitations of Correlation-Based Interpretability}

Correlation-based explanations:
\begin{itemize}
    \item are sensitive to confounders and acquisition artifacts,
    \item do not support counterfactual reasoning,
    \item may conflict with established clinical knowledge.
\end{itemize}

\subsection{Structured Causal Explanations}

Causal interpretability emphasizes:
\begin{itemize}
    \item explicit representation of known risk factors,
    \item latent disease states mediating between risk and observation,
    \item separation of disease-related and acquisition-related variability.
\end{itemize}
Explanations relate predictions to clinically recognized factors rather than low-level signal patterns.

\subsection{Neuro-Symbolic Reasoning as a Functional Requirement}

Causal abstraction requires mapping signal representations to structured variables governed by clinical knowledge and guidelines. These abstractions cannot be enforced by statistical learning alone. Neuro-symbolic reasoning is therefore required: neural models extract uncertain evidence from signals, while symbolic reasoning operates on causal abstractions, enforces domain constraints, and supports counterfactual queries. This neuro-symbolic causal pipeline is illustrated in Fig.~\ref{fig:causal_dag_pipeline}.

\subsection{Toward Intervention-Aware Clinical AI}

Clinical AI systems inform decisions that themselves constitute interventions. Models lacking causal awareness risk reinforcing biases or producing misleading recommendations under changing conditions. An intervention-aware signal processing framework aligns model behavior with clinical reasoning by explicitly modeling how changes propagate through the system.

\begin{figure}[t]
    \centering
    \includegraphics[width=\linewidth]{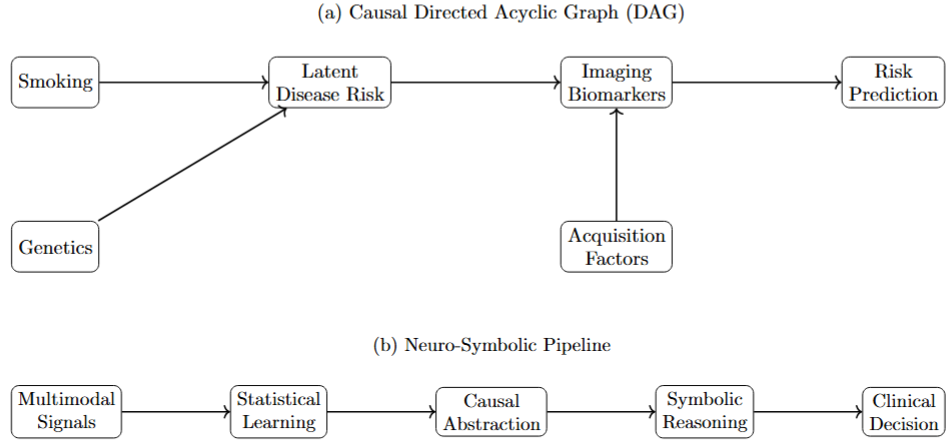}
    \caption{Causal DAG and neuro-symbolic pipeline for clinical risk prediction.
(a) Patient risk factors influence latent disease state, which generates observable biomarkers.
(b) Acquisition factors affect signal appearance without altering disease.
(c) Learned representations map signals to causal abstractions that support symbolic reasoning and intervention-aware decisions.}
\label{fig:causal_dag_pipeline}
\end{figure}

\subsection{Example: Encoding Clinical Guidelines Symbolically}
\label{subsec:guideline-example}

To illustrate how domain knowledge can be formally integrated, consider a simplified rule derived from the \textit{American Cancer Society (ACS) guidelines} for breast cancer screening~\cite{acs2023screening}:

\begin{quote}
    \textit{``Recommend annual mammography for women aged 45--54 with elevated lifetime risk (\(>20\%\)) based on family history and genetic factors.''}
\end{quote}

This guideline can be encoded symbolically using a logic-based representation. Below is a pseudocode example that could be implemented in a neuro-symbolic reasoning layer:

\begin{verbatim}
IF 
    Age $\in [45, 54]$
    AND LifetimeRisk > 0.20 
    AND (FamilyHistory = 
            TRUE OR BRCA_Mutation = TRUE)
THEN
    ScreeningRecommendation = 
            "Annual_Mammography"
    RiskCategory = "Elevated"
ELSE IF
    Age $\geq 55$
    AND LifetimeRisk $\leq 0.15$
THEN
    ScreeningRecommendation =
            "Biennial_Mammography"
    RiskCategory = "Average"
\end{verbatim}

\noindent \textbf{Integration in a Neuro-Symbolic Pipeline.} In practice, this symbolic rule set would interact with neural components as follows:
\begin{enumerate}
    \item \textit{Neural Representation Learning:} A deep learning model estimates \texttt{LifetimeRisk} from multimodal inputs (e.g., mammographic features, patient metadata).
    \item \textit{Symbolic Abstraction:} Continuous risk scores are mapped to categorical risk levels (e.g., Low/Medium/High) via thresholding or probabilistic classification.
    \item \textit{Symbolic Reasoning:} The rule engine evaluates patient-specific facts (age, family history, genetic markers from EHR) against the encoded guidelines to produce a recommendation.
\end{enumerate}

\noindent \textbf{Benefits of Symbolic Encoding.} This approach ensures:
\begin{itemize}
    \item \textit{Clinical Compliance:} Decisions adhere to established, evidence-based guidelines.
    \item \textit{Explainability:} The system can output the fired rule(s) as part of its explanation, e.g., 
        \begin{quote}
            ``Annual mammography recommended due to age 48, elevated lifetime risk (24\%), and positive family history.''
        \end{quote}
    \item \textit{Intervention-Awareness:} If screening guidelines change (e.g., updated age thresholds), only the symbolic layer requires modification---no retraining of neural feature extractors is needed, facilitating rapid policy adaptation.
\end{itemize}

Such examples demonstrate how neuro-symbolic frameworks can bridge statistical learning with clinically meaningful, auditable decision logic.

%@@@@@@@@@@@@@@@@@@----6----@@@@@@@@@@@@@@@@@@@@@@@@

\section{Neuro-Symbolic Reasoning for Causal Abstraction and Interpretable Signal Processing}

\begin{figure}[t]
    \centering
    \includegraphics[width=\linewidth]{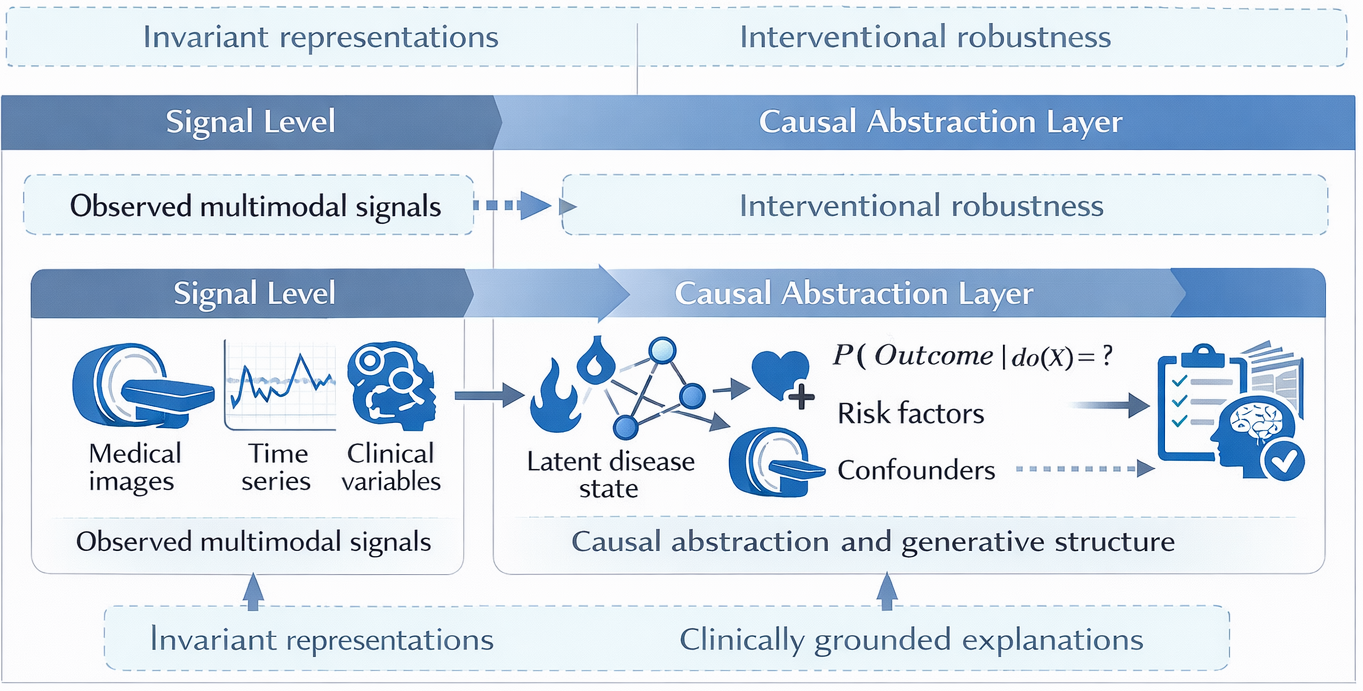}
    \caption{Proposed causal signal processing framework for clinical risk prediction.
Observed multimodal signals are mapped to latent causal abstractions that distinguish 
disease mechanisms from acquisition-related variability. 
This separation enables invariant representation learning, counterfactual reasoning, 
and intervention-aware \cite{Scholkopf2021} decision support across deployment environments.}

    \label{fig:proposed_framework}
\end{figure}

Bridging the gap between high-dimensional signal representations and clinically actionable decisions remains a central challenge in medical artificial intelligence. While statistical learning methods excel at extracting patterns from complex multimodal data, their internal representations are often opaque and difficult to align with clinical reasoning. Conversely, symbolic and rule-based systems provide transparency and explicit structure but lack the flexibility required to operate directly on raw medical signals. Neuro-symbolic systems \cite{Scholkopf2021} offer a principled framework for integrating these complementary paradigms, enabling causal, interpretable, and robust clinical decision support Fig.~\ref{fig:proposed_framework}.

\subsection{Motivation for Neuro-Symbolic Integration}

Purely data-driven models typically learn predictive associations that are sensitive to dataset-specific correlations and acquisition artifacts. In clinical settings, such sensitivity can undermine trust and reliability, particularly under distribution shifts or external interventions. Symbolic reasoning, by contrast, operates on abstract concepts such as risk factors, disease stages, and logical relationships that are familiar to clinicians and consistent with established medical guidelines.

Neuro-symbolic integration is not merely a design choice but a functional requirement for causal interpretability in clinical risk prediction. While causal neural models can learn invariant representations, they cannot by themselves enforce clinical guidelines, logical constraints, or medically meaningful relationships among risk factors. Symbolic reasoning is therefore required to encode domain knowledge, ensure consistency with established medical rules, and support structured counterfactual queries (e.g., how risk would change under altered patient behavior or screening policy). By separating statistical evidence extraction from symbolic causal reasoning, neuro-symbolic systems enable explanations and decisions that remain valid under interventions on acquisition conditions and clinical policies.

\subsection{Causal Role of Symbolic Abstractions}

Within a causal framework, symbolic abstractions serve as intermediate variables that mediate between learned signal representations and clinical decisions. These abstractions may correspond to identifiable risk factors, disease indicators, or clinically meaningful categories derived from data. By operating on such abstractions, symbolic reasoning layers facilitate causal interpretability and support counterfactual analysis.

For example, rather than attributing a prediction to low-level image features, a neuro-symbolic system can explain risk in terms of elevated disease indicators combined with known patient risk factors. Such explanations remain meaningful under interventions affecting acquisition conditions, thereby improving robustness and clinical relevance.

\subsection{From Statistical Learning to Symbolic Reasoning}

Formally, learned representations are mapped to symbolic causal abstractions
\begin{equation}
S = \eta(Z),
\end{equation}
and final decisions are produced via symbolic reasoning
\begin{equation}
\hat{Y} = \arg\max_y P(y \mid S, \mathcal{K}),
\end{equation}
where $\mathcal{K}$ encodes clinical knowledge, guidelines, and logical constraints.
This separation enables interpretable and intervention-aware decision making.

A typical neuro-symbolic pipeline for clinical risk prediction can be conceptualized as follows:

\begin{itemize}
    \item \textbf{Statistical representation learning} extracts features from multimodal signals such as medical images, physiological measurements, clinical variables, or molecular profiles.
    \item \textbf{Causal abstraction} maps these representations to interpretable variables that reflect underlying disease mechanisms or clinically relevant factors.
    \item \textbf{Symbolic reasoning} applies logical constraints, domain knowledge, and clinical guidelines to produce consistent and explainable decisions.
\end{itemize}

This modular structure enables each component to be evaluated and validated independently, supporting transparency, interpretability, and regulatory compliance in clinical applications

%\begin{boxedtext}
\subsubsection*{\textbf{Running Example: Lung Cancer Risk Prediction}}

Consider a clinical system designed to estimate lung cancer risk from low-dose chest CT scans combined with patient metadata. A conventional correlational model learns a direct mapping from image features and clinical variables to a risk score. In practice, such a model may inadvertently exploit correlations between scanner-specific artifacts and disease prevalence if these artifacts co-occur with positive cases in the training dataset.

From a causal perspective, lung cancer risk is driven by latent biological mechanisms influenced by known causal factors such as smoking history, age, and genetic predisposition. These factors generate observable manifestations, including pulmonary nodules and tissue texture patterns in CT images. Acquisition conditions, such as scanner type or reconstruction kernel, affect the appearance of these signals without altering the underlying disease process.

In a neuro-symbolic framework, statistical representation learning extracts imaging biomarkers from CT scans and integrates them with patient variables. These representations are then mapped to causal abstractions, such as estimated disease burden or progression stage. A symbolic reasoning layer combines these abstractions with known risk factors using clinically motivated rules or probabilistic logic to produce a final risk estimate. Explanations are thus expressed in terms of interpretable causal variables (e.g., elevated nodule burden combined with heavy smoking history) rather than pixel-level features.
%\end{boxedtext}

\begin{table}[t]
\centering
\caption{Comparison of correlation-based, causal, and neuro-symbolic approaches to clinical risk prediction.}
\begin{tabular}{p{2.2cm} p{1.5cm} p{1.6cm} p{1.8cm}}
\hline
\textbf{Aspect} & \textbf{Correlation-Based} & \textbf{Causal} & \textbf{Neuro-Symbolic} \\
\hline
Input modeling & Raw signals treated as predictive features. & Signals modeled as effects of latent causes. & Signals mapped to causal abstractions and symbolic variables. \\
Robustness & Sensitive to dataset-specific correlations. & Seeks invariance under non-causal interventions. & Combines invariance with rule-based constraints. \\
Interpretability & Post hoc feature attributions. & Explanations in terms of causal variables. & Explanations in terms of causal variables and symbolic rules. \\
Counterfactuals & Not supported. & Supported via causal models. & Supported via symbolic and causal reasoning. \\
Clinical alignment & Weak. & Moderate. & Strong. \\
Deployment behavior & Brittle under distribution shift. & More stable under intervention. & Stable and auditable. \\
\hline
\end{tabular}
\end{table}

\subsection{Intervention-Aware Decision Making}

An important advantage of neuro-symbolic reasoning is its compatibility with intervention-aware decision making. Symbolic layers can explicitly encode assumptions about which variables are causal and which are confounded, thereby enabling reasoning under hypothetical changes such as modified screening policies, altered imaging protocols, or different patient behaviors. As a result, decisions become less sensitive to spurious correlations and more aligned with causal understanding.

Moreover, symbolic reasoning supports uncertainty-aware outputs by exposing ambiguities in causal structure or data availability. Such transparency is essential in high-stakes clinical environments, where automated decisions must be interpretable, justifiable, and, when necessary, deferred to human expertise.

\subsection{Implications for Signal Processing Research}

From a signal processing perspective, neuro-symbolic reasoning represents a shift from purely end-to-end optimization toward structured and interpretable pipelines. This shift encourages the design of representations that are not only predictive but also causally meaningful. It also opens new research directions in signal transformations, multimodal fusion strategies, and learning objectives that facilitate symbolic abstraction and causal reasoning.

By positioning neuro-symbolic methods as a bridge from signals to decisions, this framework aligns with the broader objective of moving from correlation-based analysis to causal inference within the signal processing community.

\subsection{Computational Considerations in Neuro-Symbolic Pipelines}
\label{subsec:computational}

While neuro-symbolic integration offers significant benefits in terms of interpretability and causal consistency, it introduces unique computational challenges that must be considered for clinical deployment. The hybrid pipeline typically involves three primary stages, each with distinct complexity profiles:

\begin{enumerate}[label=(\roman*)]
    \item \textbf{Representation Learning Complexity:} Deep learning models for feature extraction---such as convolutional neural networks for medical images or transformers for clinical text---remain computationally intensive, particularly in high-dimensional multimodal settings. Training these models requires substantial labeled data and GPU resources, while inference latency must meet clinical real-time constraints in applications such as critical care monitoring or intraoperative imaging.
    
    \item \textbf{Symbolic Abstraction Overhead:} Mapping continuous neural embeddings to discrete symbolic abstractions (e.g., via clustering, rule-based thresholds, or attention mechanisms) adds additional inference steps. This abstraction layer must balance precision with computational efficiency, especially when handling probabilistic or uncertain mappings.
    
    \item \textbf{Symbolic Reasoning Scalability:} Logical inference, constraint satisfaction, and rule-based reasoning may become computationally expensive as the number of clinical guidelines, risk factors, and patient-specific constraints grows. In the worst case, symbolic reasoning can exhibit combinatorial complexity, though clinical rule sets are often intentionally designed to be tractable.
    
    \item \textbf{Integration and Optimization Costs:} Joint training or fine-tuning of neural and symbolic components often requires specialized optimization techniques, such as neuro-symbolic differentiable programming or expectation-maximization schemes. These approaches can be less scalable than standard end-to-end deep learning and may require careful hyperparameter tuning.
\end{enumerate}

\noindent \textbf{Mitigation Strategies.} Future research could explore several directions to improve scalability:
\begin{itemize}
    \item Approximate neuro-symbolic inference methods, such as relaxed logical constraints or probabilistic soft logic, which trade exactness for efficiency.
    \item Modular and lazy evaluation strategies that activate only relevant symbolic rules per clinical case, reducing unnecessary computation.
    \item Hardware-aware design tailored to deployment settings (e.g., edge devices for point-of-care applications versus cloud-based batch processing).
    \item Explicit trade-off analyses between interpretability gains and computational overhead, enabling designers to choose appropriate operating points for given clinical scenarios.
\end{itemize}

These considerations are essential for translating neuro-symbolic causal frameworks from research prototypes into scalable, clinically viable systems.

\section{Illustrative Use Case: Risk Prediction in Medical System}

Computer-Aided Diagnosis (CAD) schemes are algorithms/schemes that provide a second opinion for radiologists in their assessment of medical images \cite{Doi2007}. Their use in different imaging modalities has gained popularity over the past few decades. More recently, CAD schemes have transitioned from detection and diagnosis to disease classification (e.g., cancer) and, subsequently, to risk prediction, forecasting future cancer occurrence in sequential medical imaging data \cite{Mikhael2023}\cite{Yala2021}.

These risk prediction schemes are very important, as they facilitate personalized screening regimens and precision medicine, both of which can be highly beneficial to the public. For example, the Massachusetts Institute of Technology (MIT) released the Mirai breast cancer risk model \cite{Yala2021}, which can predict short- and long-term cancer risk based on a woman's yearly screening mammogram. 

Using these risk models, individualized screening regimens can be prescribed for individual women rather than a blanket/fixed screening regimen, e.g., making it mandatory for all women to undergo yearly or biannual screening mammograms. This reduces anxiety in women due to false positives on their mammograms, and also reduces screening costs and unnecessary radiation exposure to women due to unnecessary mammograms.  

Consider a system designed to estimate lung cancer risk from Computed Tomography (CT) scans and clinical data. A correlational model may learn associations between scanner artifacts and cancer labels if such artifacts are correlated with disease prevalence in the training dataset. A causal model, by contrast, explicitly represents smoking history and genetic predisposition as causes of latent cancer risk, which in turn generates observable nodules and texture patterns.

Under a scanner change, the correlational model may fail, whereas the causal model remains stable because it conditions on disease-related variables rather than acquisition artifacts. Explanations from the causal model can attribute risk to smoking behavior and observed biomarkers rather than to irrelevant imaging patterns.

\subsection{Additional Clinical Exemplars}
\label{subsec:exemplars}

\subsubsection*{ECG-Based Cardiovascular Risk Prediction}
Cardiac arrhythmia detection and cardiovascular risk stratification from electrocardiograms (ECGs) provide a compelling case for causal signal processing. ECGs are influenced by:
\begin{itemize}
    \item \textbf{Causal Sources:} Myocardial ischemia, structural heart disease, genetic channelopathies.
    \item \textbf{Acquisition Confounders:} Electrode placement, skin impedance, device manufacturer, motion artifacts.
    \item \textbf{Temporal Confounders:} Medication effects, electrolyte imbalances, autonomic tone variations.
\end{itemize}
A correlational model might learn to associate specific ECG device artifacts with certain arrhythmias if training data is imbalanced across devices. A causal approach would seek representations invariant to acquisition while preserving disease-related features like ST-segment morphology or QT interval. Neuro-symbolic integration could encode clinical rules (e.g., \textit{``Prolonged QTc $>$ 500ms with syncope suggests high risk of torsades''}) alongside learned embeddings.

\subsubsection*{Diabetic Retinopathy Screening from Fundus Images}
Retinal fundus photography is widely used for diabetic retinopathy (DR) screening. The causal structure includes:
\begin{itemize}
    \item \textbf{Causal Factors:} Diabetes duration, glycemic control, hypertension.
    \item \textbf{Latent Disease State:} Microvascular damage, neovascularization, retinal ischemia.
    \item \textbf{Acquisition Variables:} Camera model, pupil dilation, image resolution, illumination.
\end{itemize}
Correlational deep learning models may exploit dataset-specific correlations (e.g., certain camera brands used in clinics with higher DR prevalence). A causal model would disentangle true DR biomarkers (microaneurysms, hemorrhages) from acquisition artifacts. Symbolic reasoning could integrate guidelines from the International Clinical Diabetic Retinopathy Disease Severity Scale, ensuring that risk predictions align with clinically established progression pathways.

\subsubsection*{Comparative Insights}
These examples share a common causal motif: 
\begin{enumerate}
    \item A latent disease state generates observable signals.
    \item Acquisition processes act as confounding channels.
    \item Clinical guidelines provide symbolic constraints on decision-making.
\end{enumerate}
This universality reinforces the value of the proposed causal signal processing framework across diverse clinical domains.

\section{Implications for Signal Processing Methodology}

A causal signal processing framework has several methodological implications:
\begin{itemize}
    \item \textbf{Representation learning} should emphasize invariance to non-causal variability.
    \item \textbf{Multimodal fusion} should respect causal roles of modalities rather than treating them as interchangeable.
    \item \textbf{Evaluation} should include stress tests under simulated or real interventions.
    \item \textbf{Uncertainty modeling} should reflect ambiguity in causal structure and data availability.
\end{itemize}

\begin{boxedtext}
\textbf{Box 3: Implications for Signal Processing Research}

A causal perspective suggests several directions for future signal processing methodology:

\begin{itemize}
    \item \textbf{Causality-aware representation learning:} Feature extraction should emphasize invariance to non-causal variability induced by acquisition conditions.
    
    \item \textbf{Interventional data modeling:} Distribution shifts can be interpreted as interventions, enabling principled stress-testing of learned models.
    
    \item \textbf{Causal multimodal fusion:} Different signal modalities should be integrated according to their causal roles (cause, mediator, or effect).
    
    \item \textbf{Robust evaluation protocols:} Model validation should assess stability under simulated or real interventions rather than relying solely on i.i.d.\ test data.
\end{itemize}
\end{boxedtext}

\section{Open Challenges and Future Research Directions}
\label{sec:challenges}
While causal and neuro-symbolic approaches offer a promising pathway toward robust and interpretable clinical AI, significant open challenges remain. Addressing these challenges will require advances across signal processing, causal inference, machine learning, and clinical science. In this section, we outline key research directions that are critical for translating causal signal processing frameworks into reliable real-world medical systems.

\subsection{Causal Identifiability in Observational Medical Data}

\paragraph{Identifiability} We introduce a distributional and invariance-based formulation without Pearl-Style SCMs

\subsubsection{Model Assumptions (Non-SCM)}

Let:
\begin{itemize}
    \item $D \in \mathcal{D}$: latent disease-related variable,
    \item $A \in \mathcal{A}$: acquisition-related nuisance variable,
    \item $X \in \mathcal{X}$: observed biomedical signal,
    \item $Y \in \mathcal{Y}$: clinical outcome.
\end{itemize}

Assume the joint distribution factorizes as
\begin{equation}
P(D, A, X, Y) = P(D)\, P(A)\, P(X \mid D, A)\, P(Y \mid D).
\label{eq:factorization}
\end{equation}

This factorization encodes that:
\begin{itemize}
    \item $D$ is the sole common cause of $X$ and $Y$,
    \item $A$ influences observations but not outcomes,
    \item $Y$ depends on $X$ only through $D$.
\end{itemize}

\subsubsection{Environments as Distribution Shifts}

Let $\mathcal{E}$ be a set of environments such that for each $e \in \mathcal{E}$,
\begin{equation}
P_e(D) = P(D), 
\qquad
P_e(Y \mid D) = P(Y \mid D),
\end{equation}
while
\begin{equation}
P_e(X \mid D, A) \neq P_{e'}(X \mid D, A)
\quad \text{for some } e \neq e'.
\label{eq:env_shift}
\end{equation}

Thus, environments differ only in acquisition-related variability.

\subsubsection{Invariant Representation Definition}

Let $Z = \phi(X)$ be a learned representation.
We say that $Z$ is \emph{predictively invariant} if
\begin{equation}
P_e(Y \mid Z) = P_{e'}(Y \mid Z),
\qquad \forall e, e' \in \mathcal{E}.
\label{eq:predictive_invariance}
\end{equation}

\subsubsection{Identifiability via Predictive Invariance}

\begin{theorem}
Under assumptions \eqref{eq:factorization}--\eqref{eq:predictive_invariance}, if
$Z = \phi(X)$ satisfies predictive invariance across environments, then there exists
a measurable function $\psi$ such that
\begin{equation}
Z = \psi(D) \quad \text{a.s.},
\label{eq:identifiability}
\end{equation}
up to an invertible transformation.
\end{theorem}

\begin{proof}
From the assumed factorization, we have the conditional independence
\begin{equation}
Y \perp\!\!\!\perp X \mid D.
\label{eq:ci}
\end{equation}

For any environment $e$,
\begin{equation}
P_e(Y \mid Z)
= \int P(Y \mid D)\, P_e(D \mid Z)\, dD.
\label{eq:mixture}
\end{equation}

By Bayes' rule,
\begin{equation}
P_e(D \mid Z) \propto P_e(Z \mid D)\, P(D).
\label{eq:bayes}
\end{equation}

Assume, toward contradiction, that $Z$ is not a function of $D$ alone. Then there exist
environments $e, e'$ such that
\begin{equation}
P_e(Z \mid D) \neq P_{e'}(Z \mid D).
\label{eq:contradiction}
\end{equation}

Since $P(D)$ is invariant across environments, \eqref{eq:bayes} implies
\begin{equation}
P_e(D \mid Z) \neq P_{e'}(D \mid Z),
\end{equation}
which, substituted into \eqref{eq:mixture}, yields
\begin{equation}
P_e(Y \mid Z) \neq P_{e'}(Y \mid Z),
\end{equation}
contradicting predictive invariance \eqref{eq:predictive_invariance}.

Hence,
\begin{equation}
P_e(Z \mid D) = P_{e'}(Z \mid D),
\qquad \forall e, e' \in \mathcal{E},
\end{equation}
which implies that $Z$ cannot depend on $A$ except through $D$. Therefore,
\eqref{eq:identifiability} holds up to an invertible transformation.
\end{proof}

\subsubsection{Corollaries}

\begin{corollary}[Robustness Without Causal Graphs]
Any predictor $h(Z)$ trained on a predictively invariant representation satisfies
\begin{equation}
P_e(Y \mid Z) = P(Y \mid Z),
\qquad \forall e \in \mathcal{E},
\end{equation}
and is therefore robust to acquisition-induced distribution shifts.
\end{corollary}

\begin{corollary}[Implicit Counterfactual Interpretation]
For any intervention that modifies the data distribution while preserving
\eqref{eq:factorization},
\begin{equation}
\Delta A:
\quad
P(X \mid D, A) \mapsto P'(X \mid D, A),
\end{equation}
predictions remain unchanged, whereas modifications of $P(D)$ induce changes in
$P(Y)$.
\end{corollary}

%---------------------------------------

\iffalse

\begin{proposition}[Identifiability under Non-Causal Interventions]
Consider the structural causal model defined in (1)--(3), where acquisition variables
$A$ affect the observations $X$ but not the latent disease state $D$ nor the outcome $Y$.
Let $Z = \phi(X)$ be a learned representation satisfying the invariance condition
\begin{equation}
P_e(Y \mid Z) = P_{e'}(Y \mid Z), \quad \forall e,e' \in \mathcal{E},
\end{equation}
for environments $e$ corresponding to interventions on $A$.
Then $Z$ is a function of $D$ alone, i.e.,
\begin{equation}
Z = \psi(D)
\end{equation}
almost surely, up to an invertible transformation.
\end{proposition}

\begin{proof}[Proof Sketch]
Under the assumed SCM, interventions on $A$ modify the conditional distribution
$P(X \mid D)$ while leaving $P(Y \mid D)$ invariant. Suppose that $Z = \phi(X)$
retains information about $A$ beyond what is mediated by $D$. Then there exist
environments $e$ and $e'$ such that
\[
P_e(Z \mid D) \neq P_{e'}(Z \mid D),
\]
which induces environment-dependent conditionals $P_e(Y \mid Z)$ through the
dependence of $Z$ on $A$. This contradicts the assumed invariance of $P(Y \mid Z)$
across environments.

Therefore, $Z$ must be conditionally independent of $A$ given $D$, implying that
$Z$ is a function of $D$ alone, up to an invertible transformation. Such equivalence
is sufficient for robust prediction and causal interpretability under non-causal
interventions.
\end{proof}
\fi

%---------------------------------------

A fundamental challenge in medical AI is that most available data are observational rather than experimental. As a result, causal relationships among variables are often not identifiable without strong assumptions. In clinical settings, randomized controlled trials are costly, time-consuming, and frequently infeasible.

Future research must focus on developing signal processing and representation learning methods that facilitate causal identifiability under realistic assumptions. This includes leveraging domain knowledge, exploiting natural experiments, and designing representations that separate disease-related variability from confounding effects. Identifiability-aware modeling is essential to ensure that learned causal relationships reflect true disease mechanisms rather than dataset-specific artifacts.

\paragraph{Identifiability Under Acquisition Interventions.}
Consider interventions on acquisition variables $do(A=a)$ that alter signal appearance
without affecting disease mechanisms.
Under the SCM in (1)–(3), such interventions change $P(X \mid D)$ while leaving
$P(Y \mid D)$ invariant.
Any representation $Z=\phi(X)$ that retains acquisition-specific information will induce
environment-dependent conditionals $P(Y \mid Z, e)$.
Enforcing invariance across environments therefore eliminates components of $X$
that vary with $A$ but not with $D$, yielding
\begin{equation}
Z = \psi(D),
\end{equation}
up to an invertible transformation.
Thus, disease-related representations are identifiable up to equivalence under
non-causal interventions, which is sufficient for robust prediction and interpretation.

\subsection{Validation of Causal Models Without Interventions}

Evaluating causal models remains challenging when ground-truth interventions are unavailable. Standard performance metrics, such as accuracy or area under the curve, do not capture robustness under intervention or causal correctness.

New evaluation protocols are needed that stress-test models under simulated or proxy interventions, such as cross-institutional validation, acquisition perturbations, or policy-driven data shifts. Signal processing research can contribute by developing principled methods for generating and analyzing such perturbations while preserving clinical realism.

\subsection{Multimodal Fusion Under Causal Constraints}

Medical AI increasingly relies on heterogeneous data sources, including imaging, clinical records, genomics, and longitudinal measurements. While multimodal fusion has shown promise, most existing methods treat modalities as interchangeable inputs rather than as variables with distinct causal roles.

A key research direction is the development of fusion strategies that respect causal structure, distinguishing between causes, mediators, and confounders across modalities. Such approaches can improve interpretability, reduce overfitting, and enable more reliable decision-making under missing or corrupted data.

\subsection{Uncertainty Quantification and Causal Ambiguity}

Uncertainty is inherent in medical decision-making, particularly when causal relationships are only partially identifiable. Models that provide point predictions without uncertainty estimates may give a false sense of confidence and lead to inappropriate clinical actions.

Future work should integrate uncertainty quantification with causal modeling, explicitly representing ambiguity arising from limited data, confounding, or model assumptions. Neuro-symbolic frameworks are particularly well suited to this task, as they can expose uncertainty at both the statistical and symbolic levels.

\subsection{Human--AI Collaboration and Clinical Integration}

Ultimately, clinical AI systems are designed to support, not replace, human decision-makers. Effective integration requires that model outputs and explanations align with clinical workflows, guidelines, and reasoning processes.

Research is needed to understand how causal explanations influence clinician trust, decision-making, and accountability. Signal processing researchers can contribute by designing systems that communicate causal assumptions transparently and support interactive exploration of alternative scenarios and interventions.

\subsection{Regulatory and Ethical Considerations}

Causal interpretability has important implications for regulatory approval and ethical deployment of medical AI. Transparent reasoning about why decisions are made, and under what assumptions they remain valid, is essential for meeting emerging regulatory standards.

Future work should explore how causal signal processing frameworks can support auditability, bias detection, and fairness analysis. By making causal assumptions explicit, such frameworks offer a pathway toward more responsible and trustworthy AI systems.

\section{Conclusion}
\label{sec:conclusion}

This work advocates a fundamental shift in medical artificial intelligence from correlation-driven prediction toward causal, intervention-aware signal processing frameworks. By treating multimodal biomedical signals as observations generated by latent causal mechanisms, rather than as isolated inputs for pattern recognition, such frameworks provide a principled foundation for clinical inference and decision support.

Integrating symbolic reasoning with statistical learning further enables models that are not only accurate but also interpretable, uncertainty-aware, and aligned with clinical reasoning processes. From this perspective, neuro-symbolic and causal abstractions support counterfactual analysis, improve robustness under distribution shifts, and facilitate generalization across patient populations and clinical settings.

Taken together, causal signal processing offers a unifying framework for addressing interpretability, robustness, and trust in medical AI. We believe this paradigm represents a critical step toward the development of reliable and transparent clinical decision-support systems and constitutes a timely and fertile research direction for the signal processing community.

\appendix
\section*{Appendix A: Resources for Practitioners}
%\section{Resources for Practitioners}
\label{sec:appendix_resources}

This appendix provides a curated list of software libraries, datasets, and foundational papers to facilitate the implementation of causal and neuro-symbolic signal processing pipelines for clinical risk prediction.

\subsection{Software Libraries}
\label{sec:appendix_software}

\paragraph{Causal Inference}
\begin{itemize}
  \item \textbf{DoWhy (Python)}: End‑to‑end causal inference library supporting structural causal models, effect identification, estimation, and robustness checks within a unified four‑step workflow (model–identify–estimate–refute)~\cite{doWhy}.

  \item \textbf{CausalNex (Python)}: Library for Bayesian network structure learning and inference, aimed at causal discovery and reasoning in real‑world decision pipelines~\cite{causalNex}.

  \item \textbf{gCastle (Python)}: Toolkit for causal structure learning from observational data, supporting multiple state‑of‑the‑art graph discovery algorithms~\cite{gCastle}.
\end{itemize}

\paragraph{Neuro-Symbolic AI}
\begin{itemize}
  \item \textbf{DeepProbLog (Python)}: Neural probabilistic logic programming framework that integrates deep neural networks as predicates into ProbLog, enabling end‑to‑end training over symbolic and subsymbolic components~\cite{deepProbLog}.

  \item \textbf{Neural Theorem Prover (PyTorch)}: Differentiable reasoning architecture that learns to prove logical statements using gradient‑based optimization over symbolic knowledge~\cite{neuralTheoremProver}.

  \item \textbf{Logic Tensor Networks (TensorFlow)}: Framework that combines first‑order logic with neural networks by relaxing logical constraints into differentiable loss terms for learning with soft logical rules~\cite{ltensorNetworks}.
\end{itemize}

\paragraph{Invariant Representation Learning}
\begin{itemize}
  \item \textbf{InvariantRiskMinimization (PyTorch)}: Reference implementation of Invariant Risk Minimization (IRM) for learning predictors whose optimality is stable across environments, used for domain generalization benchmarks~\cite{irm}.

  \item \textbf{DomainBed (Python)}: Benchmark suite for domain generalization that implements a wide range of algorithms (including IRM) with standardized evaluation across multiple datasets and shifts~\cite{domainBed}.
\end{itemize}

\subsection{Public Medical Datasets}
\label{sec:appendix_datasets}

\paragraph{Lung Cancer Imaging}
\begin{itemize}
  \item \textbf{NLST (National Lung Screening Trial)}: Large‑scale low‑dose chest CT screening study with longitudinal follow‑up and mortality outcomes for lung cancer risk modeling~\cite{NLST}.

  \item \textbf{LIDC‑IDRI}: Public lung CT database with expert‑annotated nodules (location, contours, and subjective ratings), widely used for nodule detection and malignancy prediction~\cite{LIDC}.
\end{itemize}

\paragraph{Cardiovascular and Multimodal Clinical Signals}
\begin{itemize}
  \item \textbf{PhysioNet CinC Challenges}: Annual Computing in Cardiology Challenge datasets (e.g., ECG for arrhythmia classification, mortality prediction, and physiological waveform analysis), hosted via PhysioNet~\cite{physionet}.

  \item \textbf{MIMIC‑III}: De‑identified ICU database linking time series (vital signs, waveforms), laboratory results, interventions, and clinical notes for outcome and risk prediction tasks~\cite{mimicIII}.

  \item \textbf{EyePACS}: Fundus photograph dataset for diabetic retinopathy grading, often used as a benchmark for retinal image‑based risk stratification~\cite{eyePACS}.

  \item \textbf{UK Biobank (retinal subset)}: Large population cohort with retinal imaging (OCT and fundus) and rich systemic phenotypes, enabling studies that connect ocular biomarkers with cardiovascular and systemic risk~\cite{ukb}.
\end{itemize}

\subsection{Foundational Readings}
\label{sec:appendix_readings}

\paragraph{Causal Inference}
\begin{itemize}
  \item J.~Pearl, \emph{Causality: Models, Reasoning, and Inference}. Cambridge University Press, 2009~\cite{Pearl2009}.

  \item J.~Peters, D.~Janzing, and B.~Schölkopf, \emph{Elements of Causal Inference: Foundations and Learning Algorithms}. MIT Press, 2017~\cite{peters}.
\end{itemize}

\paragraph{Causal Representation Learning}
\begin{itemize}
  \item B.~Schölkopf et al., ``Towards causal representation learning,'' \emph{Proceedings of the IEEE}, vol.~109, no.~5, pp.~655--675, 2021~\cite{scholkopf21}.

  \item Y.~Bengio et al., ``A meta‑transfer objective for learning to disentangle causal mechanisms,'' in \emph{International Conference on Learning Representations (ICLR)}, 2019~\cite{bengio19}.
\end{itemize}

\paragraph{Neuro-Symbolic AI}
\begin{itemize}
  \item M.~K.~Sarker et al., ``Neuro‑symbolic artificial intelligence: current trends,'' in \emph{AI Communications} / \emph{Artificial Intelligence Review}, 2021~\cite{sarker21}.

  \item R.~Riegel et al., ``Logical neural networks,'' in \emph{Proceedings of the AAAI Conference on Artificial Intelligence}, 2020~\cite{riegel20}.
\end{itemize}

\section*{Appendix B: Glossary of Causal Terminology}
\label{app:glossary}

\begin{itemize}
    \item{Structural Causal Model (SCM)} A tuple $(U, V, F)$ where $U$ are exogenous variables, $V$ are endogenous variables, and $F$ is a set of structural equations determining each $V_i$ from its parents and $U$.
    
    \item{Intervention ($do$-operator)} Notation $do(X=x)$ represents an external action that sets variable $X$ to value $x$, breaking its usual causal dependencies. For example, $do(A=\text{"new scanner"})$ models changing acquisition hardware.
    
    \item{Counterfactual} A statement about what would have happened under different conditions, e.g., ``What would the risk score be if the patient had never smoked?'' Formally, $Y_{X=x}(u)$ denotes the value of $Y$ had $X$ been $x$ for unit $u$.
    
    \item{Confounder} A variable that influences both the treatment (or exposure) and the outcome, creating a spurious association. In medical imaging, scanner type may confound the relationship between image features and disease.
    
    \item{Mediator} A variable that lies on the causal path between treatment and outcome. For example, disease biomarkers mediate the effect of risk factors on clinical outcomes.
    
    \item{Causal Graph} A directed acyclic graph (DAG) where nodes represent variables and edges represent direct causal relationships.
    
    \item{Invariance} Property of a relationship that remains stable under interventions. A representation $Z$ is invariant to acquisition changes if $P(Y|Z)$ remains constant across scanners.
    
    \item{Identifiability} Whether causal quantities (e.g., treatment effects) can be uniquely estimated from observational data given assumptions.
    
    \item{Backdoor Criterion} A set of variables that, when conditioned on, blocks all backdoor paths between treatment and outcome, allowing causal effect estimation.
    
    \item{Frontdoor Criterion} An alternative identification strategy using mediators when confounders are unobserved.
\end{itemize}

\section*{Authors}
\textbf{Surajit Das} is a PhD researcher in Artificial Intelligence and Mathematics at ITMO University with industry experience in software engineering, data analytics, and AI research. He holds an MSc in Artificial Intelligence (distinction) from Liverpool John Moores University, UK, and an MSc in Applied Mathematics from the Moscow Institute of Physics and Technology, Russia. His research focuses on uncertainty quantification, sequential decision-making, and statistical modeling for medical signal processing. He has authored peer-reviewed publications in Q1/Q2 journals and presented at Scopus-indexed international conferences (e.g., IEEE DESE). He serves as a verified peer reviewer for Web of Science–indexed journals (Q1, Q2) and has professional experience with IBM, Tata Consultancy Services, VISA (UK), and the Indian Statistical Institute, contributing to applied AI projects in healthcare and computer vision. He is a recipient of the Russian ``Open Doors'' Olympiad Scholarship in Mathematics and Artificial Intelligence and the Computer Society of India ``Young IT Professional Award.''\\

\textbf{Maxine Tan} is a Senior Lecturer in the School of Engineering at Monash University Malaysia. Her research focuses on quantitative medical image analysis and computer-aided diagnosis for cancer risk prediction, prognosis, and treatment response, particularly in breast, lung, and ovarian cancers. Her work includes deep learning and radiomics applied to CT and MRI imaging. She has secured multiple local and international research grants and published more than 50 peer-reviewed articles in leading medical imaging journals, including \emph{Medical Image Analysis}, \emph{IEEE Transactions on Medical Imaging}, \emph{Pattern Recognition}, \emph{Breast Cancer Research}, \emph{Journal of Magnetic Resonance Imaging}, \emph{Artificial Intelligence in Medicine}, and \emph{IEEE Transactions on Biomedical Engineering}. Her contributions have received several recognitions, including the SPIE Medical Imaging Cum Laude (Best Poster) Award, an Honorable Mention Award in Computer-Aided Diagnosis and Digital Pathology, and the Underwriters Laboratories ASEAN–U.S. Science Prize for Women (Mid-Career Scientist category).

\bibliographystyle{IEEEtran}

\end{document}